\documentclass[11pt]{article}
\usepackage{fullpage}
\usepackage{amsmath}
\usepackage{amsfonts}
\usepackage{amssymb}
\usepackage{amsthm}
\usepackage{graphicx}
\usepackage{mathtools}
\usepackage{mleftright}
\usepackage{cancel}

\newtheorem{theorem}{Theorem}

\newtheorem{lemma}{Lemma}

\newtheorem{definition}{Definition}

\title{Is competitive online paging an artifact?}
\author{Enoch Peserico\footnote{Università degli Studi di Padova.} \and Michele Scquizzato\footnote{Università degli Studi di Padova, \texttt{scquizza@math.unipd.it}. Supported in part by the Italian National Center for HPC, Big Data, and Quantum Computing.}}
\date{}

\begin{document}

\maketitle

\begin{abstract}
In any real system a newly computed datum begins its existence in the processor rather than in external memory, and thus does not inevitably incur a cold miss. This was captured by early I/O models, but not by the Sleator-Tarjan one that has come to underpin competitive analysis of paging. If one corrects the Sleator-Tarjan model by charging no cost for the first access to newly computed data, optimal offline algorithms such as LFD remain optimal, but no online paging algorithm can be competitive, even if randomized, even with arbitrary resource augmentation, even against request sequences that are not tailored against it but are instead representative of widely used computational techniques. The proofs are simple, and appear robust against any reasonable assumption/model adjustment, including virtually all tools developed to make competitive analysis less pessimistic.

In other words, while competitive analysis does predict the good performance exhibited in practice by online paging algorithms such as LRU, these predictions seem just a fortuitous artifact of an incorrect assumption that has crept into the underlying model several decades ago. And there are implications beyond paging, too: for example, the same issue undermines the Ideal Cache model on which the popular Cache-Oblivious and Cache-Adaptive algorithmic frameworks are based.
\end{abstract}

\section{Introduction}
\label{sec:intro}
This work analyses the impact of a subtle flaw that has crept into the theoretical analysis of paging several decades ago: the implicit assumption that essentially every datum, even if newly computed, begins its existence in external memory and thus inevitably incurs a cold miss. After Section~\ref{sec:paging} briefly reviews paging within the competitive analysis framework, Section~\ref{sec:zeroin} re-examines the aforementioned assumption, how and when it entered the literature, and how instead one should simply discount the first access to any page in a ``zero-in'' set. Sections~\ref{sec:offline} and \ref{sec:online} show that, if we remove the assumption and allow zero-in pages, offline algorithms such as LFD remain optimal regardless of the zero-in page set, but no online algorithm can be competitive, even if randomized, even with arbitrary resource augmentation, even when faced with natural request sequences that are not tailored against it and that are in fact representative of widely used computational techniques. Section~\ref{sec:robust} examines the resilience of this very negative result to various refinements of competitive analysis. Only Full-Cost analysis (where non-fault accesses cost a fraction $\epsilon>0$ of faults) offers a marginally more optimistic outlook; Section~\ref{sec:fullcost} provides a tight characterization of the $(h,k)-$competitive ratio achievable. Section~\ref{sec:out} briefly looks at a simple extension of the zero-in framework to deal with the cost of evictions. Finally, Section~\ref{sec:end} briefly summarizes our results and examines some of their implications, including how they undermine the Ideal Cache model on which all the Cache-Oblivious and Cache-Adaptive algorithms literature is based.

\section{Paging}
\label{sec:paging}
The memory/data storage system of computing devices is almost always organized as a hierarchy of several layers of progressively larger capacity but also higher access cost,
in terms of both time and energy; efficiently orchestrating the flow of information across different layers is crucial for performance. The most widely used theoretical model for studying this \emph{paging problem} is a two-layer system: a smaller \emph{memory} layer with a capacity of $k$ \emph{pages} (data blocks), and a larger layer of infinite capacity whose pages can only be accessed by first copying them into memory -- an operation usually termed a {\em (page) fault}. Given any sequence of pages that must be accessed in order, a paging algorithm must choose which page(s) to evict from memory, whenever a new page must be copied into it, so as to minimize the total number of faults.

The simple algorithm LFD that evicts the page accessed furthest in the future has long been known to be optimal~\cite{online}. However, paging is often studied as an {\em online} problem, i.e. an algorithm can decide evictions only on the basis of past requests. A popular framework for evaluating the performance of online paging algorithms is that of \emph{competitive analysis}~\cite{competitive}.  A paging algorithm is said to have an $(h,k)-$\emph{competitive ratio} of (no more than) $\rho$ if, for every request sequence, it incurs in expectation with a memory of size $k$ at most $\rho$ times as many faults as an optimal offline algorithm incurs with a memory of size $h\leq k$, plus a number of faults independent of the request sequence. The ratio $\frac{k}{h}$ is called the \emph{resource augmentation}. Resource augmentation and competitive ratio capture, respectively, the space and access-cost overheads incurred by an online algorithm.

Even with the worst-case approach of competitive analysis, many simple deterministic online algorithms have an $(h,k)-$competitive ratio (optimal for deterministic algorithms) of $\frac{k}{k-h+1}$~\cite{SleatorT85,online} and thus never fare worse than an optimal offline algorithm would on a system with half the memory capacity and twice the access cost. \cite{online}, \cite{onlinesurvey05} and~\cite{competitivealternatives} provide excellent surveys of the many variants of competitive analysis for the paging problem:
these include randomization~\cite{marking}, somehow limiting the choice of the adversarial request sequence~\cite{accessgraph,multifinger}, amortizing the performance evaluation over a spectrum of ``related'' sequences~\cite{averageorder,bijective} or of memory sizes~\cite{youngcachechanges}, considering pages of different size and access cost~\cite{iranicachechanges,youngcachechanges}, and accounting for the non-zero cost of non-fault requests~\cite{fullcost}.

\section{Zero-in paging}
\label{sec:zeroin}
A crucial shortcoming of the widely adopted paging model described in the previous section is the implicit assumption that every page begins its existence outside of memory, and thus incurs an inevitable fault when accessed for the first time -- a so-called ``cold miss''.

Instead, in modern operating systems
a newly allocated virtual memory page is empty, and thus there is no need to bring its (non-existing) contents into memory the first time it is accessed. The system simply maps the virtual page address range to the address range of an unused portion of physical memory, avoiding a disk read and the corresponding latency. Likewise, the first time a process accesses data produced by another process and still present in memory, there is no need to perform a disk read.
Something similar happens, rather than at the disk-memory interface, at the cache-registers interface within processors: an efficient compiler can generally recognize the first access to a variable, and generate code that does not issue a read to the corresponding memory location but instead directly accesses, at a much reduced cost, the register into which the variable will be mapped.

In a nutshell, a newly computed datum begins its existence in the processor ALUs/registers, at the very top of the memory hierarchy rather than at the bottom; so no additional time or energy must generally be spent to access it for further processing. This was indeed captured by early theoretical I/O models~\cite{HongK81}. Unfortunately, the first work casting paging within the competitive analysis framework~\cite{SleatorT85} considered paging as a special case of list update, with new list elements/pages being placed in the last (i.e. most expensive) list position rather than the first -- implicitly assuming that new data begin their existence as far away as possible from the processor ALUs/registers rather than within them. And all subsequent work on online paging has followed suit, assuming that every page in a request sequence begins its existence outside of memory and thus inevitably incurs the full cost of a fault when first accessed (with the possible exception of at most $k$ pages in memory before the first request is issued, which are too few to impact the competitive ratio and are thus safely ignored~\cite{online}). 

Our work investigates the consequences of removing this incorrect assumption, designating instead an \emph{arbitrary} set of pages in any given request sequence as ``zero-in'' and charging no cost for the initial access to any such page. Zero-in pages can represent newly computed data, as well as data inherited at the top of the memory hierarchy from a different process, including one running in parallel~\cite{shortmulticore}. We shall hereafter denote by \emph{ST paging} the paging model described in the previous section and underpinning most work on competitive online paging, and by \emph{zero-in paging} the (older) paging model that simply removes from ST paging the assumption that all pages begin their existence outside of memory.

Formally, zero-in paging assumes that any request sequence $\sigma=p_1,\dots,p_n$ is paired with a zero-in page set $Z$, revealed to any paging algorithm before the first page $p_1$ in $\sigma$;
starting with an empty memory, paging then proceeds exactly as in the ST paging, except that for each page in $Z$ the cost of accessing it for the first time is $0$ rather than $1$.\footnote{An alternative would be to reveal to the paging algorithm whether a page belongs to $Z$ or not the first time it is accessed. It is easy to prove that the two formulations are equivalent (note that we do not require all pages in $Z$ to be part of $\sigma$, so knowledge of a yet unencountered page in $Z$ provides no information on future requests); the current formulation yields a slightly cleaner model, in particular when it is extended to zero-out pages -- see Section~\ref{sec:out}.} More precisely, when servicing $\sigma$ with a memory of $k$ pages, let $M^k_0 = \emptyset$, and $M^k_i$ for $i=1,\dots, n$ the set of pages in memory immediately after $p_i$ is serviced; obviously $|M^k_i|\leq k$ and $p_i \in {M^k_i}$. Then, the cost incurred by ALG to service the $i^{th}$ request for $i=1,\dots,n$ is
\begin{equation}
\label{eqn:single_page_cost}
c_i = |(M^k_i\setminus M_{i-1}^k)\setminus (Z\cap (M_i^k\setminus \cup_{j=0}^{i-1} M_j^k))|
\end{equation}
and the total cost incurred by ALG on the request sequence $\sigma$ with zero-in page set $Z$ is $c_{ALG}^Z(\sigma)=\sum_{i=1}^n c_i$.

The definition above is very general and makes no assumptions on ALG, such that ALG be demand paging, except that any page be moved into or out of memory only once in response to a single request; so the cost to service the $i^{th}$ request equals the number of new pages in the memory since the $(i-1)^{th}$ request, $|(M^k_i\setminus M_{i-1}^k)|$, ignoring those that are zero-in and were never accessed before, i.e. those in $Z\cap (M_i^k\setminus \cup_{j=0}^{i-1} M_j^k)$. The assumption of an initially empty memory is not restrictive; to model an initial set $I$ of pages in memory, one can simply prepend to $\sigma$ a prefix with one request to each page in $|I|$, and include $I$ in $Z$ so this ``prefetching'' phase incurs no cost. Finally, note that in principle the eviction choices of ALG could depend on $Z$; if that is not the case, i.e. if, for every $\sigma$ and $k$, ALG exhibits the same memory contents $M_i^k$ at each step regardless of $Z$, we say that ALG is \emph{zero-independent}. All paging algorithms in the literature are, to the best of our knowledge, zero-independent.

\section{Offline zero-in paging}
\label{sec:offline}
It is very easy to prove that the actual zero-in page set $Z$ of a request sequence has essentially no impact on optimal offline paging. More precisely, if we take a paging algorithm that is optimal for \emph{some} zero-in page set $\bar Z$, then always running it \emph{as if} the page set \emph{were} $\bar Z$ (regardless of what it actually is) yields a paging algorithm that is optimal for \emph{every} zero-in page set $Z$ (and is obviously zero-independent too). Thus LFD, that is optimal with an empty zero-in page set and always runs as if the zero-in page set were empty, remains optimal for any zero-in page set.

Formally, denoting by ALG$_Z$ the paging algorithm making the same replacement choices that ALG would if the zero-in page set were $Z$ (note that ALG$_Z$ is always zero-independent even if ALG is not) we can prove the following:

\begin{theorem}
\label{thm:offline}
If $c^{\bar Z}_{ALG}(\sigma) \leq c^{\bar Z}_{ALG'}(\sigma)$ for all $ALG'$ and $\sigma$, then
$c^{Z}_{ALG_{\bar Z}}(\sigma) \leq c^{Z}_{ALG'}(\sigma)$ for all $ALG'$, $\sigma$, and $Z$.
\end{theorem}

\begin{proof}
Given a set of pages $Z$ and a request sequence $\sigma$, with a slight abuse of notation denote by $Z\cap\sigma$ the set of pages in both $Z$ and $\sigma$; then for any ALG, $Z, \sigma$, if ALG is zero-independent $c^Z_{ALG}(\sigma)=c^\emptyset_{ALG}(\sigma) - |Z\cap\sigma|$, since the only difference between having $Z$ rather than $\emptyset$ as a zero-in page set is that in the former case each page in $Z\cap\sigma$ costs $0$ rather than $1$ to access for the first time. Then for any pair of zero-independent algorithms ALG$'$ and ALG$''$ and any request sequence $\sigma$ and zero-in page set $Z$ we have that $c^Z_{ALG'}-c^Z_{ALG''}=c^{\emptyset}_{ALG'}-c^{\emptyset}_{ALG''}$. Hence, if for some ALG$'$, $\sigma, Z$ we had 
$c^{Z}_{ALG_{\bar Z}}(\sigma) > c^{Z}_{ALG'} (\sigma) = c^{Z}_{ALG'_Z}(\sigma)$ 
then with zero-in page set $\bar Z$ we would have
$c^{\bar Z}_{ALG}(\sigma) = c^{\bar Z}_{ALG_{\bar Z}}(\sigma) > c^{\bar Z}_{ALG'_Z}(\sigma)$, against the hypothesis.
\end{proof}

\section{Online zero-in paging}
\label{sec:online}
In the light of the previous section it may be surprising that with zero-in paging most results on the competitiveness of \emph{online} ST paging do not hold: in particular, no online paging algorithm can be competitive, even if randomized, even with arbitrary resource augmentation, even when servicing  natural request sequences not adversarially tailored against it. And the request sequences we consider are those characteristic of (either of) two widely used techniques: \emph{beam searches} and \emph{genetic algorithms}.

Beam search~\cite{beam} is a heuristic search algorithm that explores (incompletely) a tree breadth-first by expanding, at each level, only a limited number of candidate nodes, selected according to some optimization criterion. All remaining subtrees are pruned (see Figure~\ref{fig:branchbound}). Beam search is used in a wide variety of application domains, from game AIs~\cite{beamAI} to computational linguistics~\cite{beamCL}, to operations research~\cite{beamOR}.

A similar structure is exhibited by genetic algorithms, another technique with a large set of application domains, from RNA structure prediction~\cite{genRNA} to CAD~\cite{genCAD}, to software verification~\cite{genIS} -- naming just a very few. A genetic algorithm combines and/or mutates a small set of $m$ initial solutions $s_0^1, \dots, s^m_0$ to generate $M \gg m$ additional solutions $s_1^1,\dots,s_1^M$. Of these $M$, all but $m$ are discarded; the $m$ survivors, chosen according to some optimization criterion, are combined to generate a new crop of $M$ additional solutions  $s^1_2, \dots, s^M_2$. $m$ survivors are again selected and combined to obtain a third generation of $M$ new solutions, and so on. Ideally, this process allows the survivors of successive generations to gradually improve the quality of the solution.

Consider a beam search over $\ell$ levels, with  beam width  (i.e.\ maximum number of candidates selected for expansion at each level) equal to $m$, and average outdegree of these candidates equal to $\frac{M}{m}$ (so that up to $M$ nodes are explored at each level). Similarly, consider a genetic algorithm running for $\ell$ generations, with $m$ survivors and $M$ new solutions per generation. Both generate a sequence of requests in the form
\begin{align*}
\sigma_\ell(m,M) =\,\,  &\sigma_1^1, \sigma_1^2, \dots, \sigma_1^M,\\
      	& \sigma_2^1, \sigma_2^2, \dots, \sigma_2^M,\\
        & \qquad\quad\vdots\\
		& \sigma_\ell^1, \sigma_\ell^2, \dots, \sigma_\ell^M  
\end{align*}
where $\sigma_i^j$ is the subsequence that computes $s_i^j$ -- that is, in the case of beam search, the subsequence that obtains the $j^{th}$ node at level $i$ from one of the $m$ candidates for expansion at level $i-1$, and in the case of a genetic algorithm, the subsequence that obtains the $j^{th}$ solution at generation $i$ from one or more of the $m$ survivors at generation $i-1$. 

\begin{figure}[ht]
  \centering
  \includegraphics[width=0.45\linewidth]{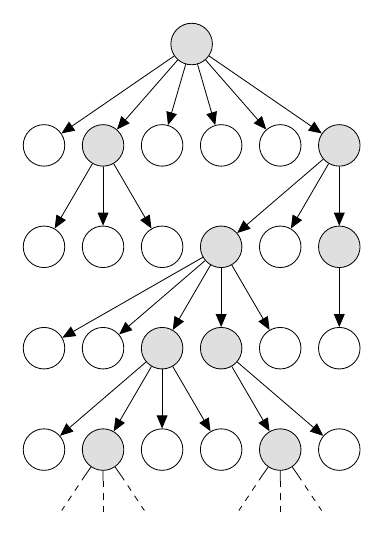}
  \caption{A beam search with beam width equal to two: only nodes in grey are expanded. This could also be seen as a genetic algorithm with two survivors per generation, and each new solution generated by (a mutation of) a single survivor.}
  \label{fig:branchbound}
\end{figure}

Assume for simplicity that each solution/node occupies exactly one page, allowing us to abuse the notation slightly and refer to ``page'' $s_i^j$; and that the $m$ \emph{parents} from which all pages of level/generation $i$ are obtained are $m$ pages $s_{i-1}^{j_i^1} \dots s_{i-1}^{j_i^m}$ chosen uniformly at random from those of level/generation $i-1$ . Note that every page other than the parents of the first level/generation (i.e. other than the inputs) is obtained via computation, and thus ``starts'' in memory and should be considered zero-in. Finally, let us make the realistic assumption that, at each level/generation, at most $m$ additional ``scratch'' pages $r_1 \dots r_m$ are used for the computation. Then, 
\begin{equation*}
\sigma_i^j \in \{s_{i-1}^{j_i^1}, \dots, s_{i-1}^{j_i^m},r_1,\dots,r_m, s_i^j \}^*
\end{equation*}
and it is easy to prove the following two lemmas:
\begin{lemma}
\label{lem:offline}
Assuming the zero-in page set includes every non-input page (i.e. it includes $p_h$ for $1\leq h\leq m$ and $s_i^j$ for $1\leq i\leq \ell$ and $1\leq j\leq M$), an offline algorithm can always service $\sigma_\ell(m,M)$ with memory $3m+1$ incurring at most $m$ faults.
\end{lemma}
\begin{proof}
The offline algorithm, when computing generation $i$, maintains in memory its parents (up to $m$), the already computed  parents of the next generation (up to $m$), the scratch pages (up to $m$), and the page currently being computed; these require at most memory $3m+1$ to store. Also, all pages in generation $i$ can be obtained from those present in memory at the end of the computation of generation $i-1$ if $i>1$, or from the parents of the first generation if $i=1$. Then, the offline algorithm running with at least $3m+1$ memory incurs at most $m$ faults (to load into memory the parents of the first generation).
\end{proof}

\begin{lemma}
\label{lem:online}
Regardless of the zero-in page set, no online algorithm can service $\sigma_\ell(m,M)$ with memory less than $\frac{M}{2}$ incurring fewer than $\frac{m(\ell-1)}{2}$ faults in expectation.
\end{lemma}
\begin{proof}
Since the request sequence is randomized and does not depend on the online algorithm, we need only prove the lemma for deterministic algorithms -- Yao's principle  then automatically yields the result for randomized algorithms. Consider the set $S_i$ of pages that must be computed at the $i^{th}$ generation, for any $i>1$. To compute $S_i$, an online algorithm must access the $m$ parents of $S_i$. Since they are chosen uniformly at random from a set of at least $M$ pages, the expected number of those parents present in memory immediately before the computation of $S_i$ begins is less than $m\cdot \frac{M/2}{M}=\frac{m}{2}$. Since those parents have already been computed immediately before $S_i$ begins, every access to one not present in memory incurs a fault (being the second or later access to that page). Thus, the online algorithm incurs in expectation more than $(\ell-1)\frac{m}{2}$ faults.
\end{proof}

The two preceding lemmas immediately yield the following:
\begin{theorem}
\label{thm:online}
Denote by $c^k_{ALG}(\sigma_\ell(m,M))$ the expected cost incurred by a generic paging algorithm ALG with memory $k$ when servicing $\sigma_\ell(m,M)$ and by $c^h_{OPT}(\sigma_\ell(m,M))$ the cost incurred by the optimal offline algorithm with memory $h$, when the zero-in page set includes
$\{p_h:1\leq h\leq m\}\cup\{s_i^j: 1\leq i \leq \ell, 1\leq j\leq M\}$). Then, for \emph{any arbitrarily large} $k$, $\alpha$ and $\beta$,
\begin{equation*}
\label{eqn:online}
c^k_{ALG} (\sigma_\ell(m,M)) >
\alpha + \beta\cdot c^{3m+1}_{OPT} (\sigma_\ell(m,M)) 
~~~~~\forall M,\ell~>~\max\{2k,(4\alpha+1),(4\beta+1)\}.
\end{equation*}
\end{theorem}

In other words, no matter what level of resource augmentation we allow an online algorithm, for any arbitrarily large additive constant $\alpha$ and multiplicative constant $\beta$, it will incur an expected cost greater than $\alpha$ plus $\beta$ times the cost of the offline optimal on all beam searches and genetic algorithm computations with sufficiently many levels/generations and sufficiently large datasets at each level/generation.

It is interesting to compare the predictions of ST paging and zero-in paging with what happens to real code running beam searches and/or genetic algorithms -- assuming the simplest but still realistic request sequences, where computing each node/solution $s_i^j$ takes just one access to a single parent, one to a scratch page, and one to $s_i^j$ itself. In practice, programmers and compilers go to great lengths to ensure that essentially the only faults incurred are those to access the inputs, either by trying to identify the parents of the next level/generation as soon as they are computed or -- when that is impossible -- by keeping each level/generation sufficiently small to fit within the available memory. In other words, they obtain an average fault rate close to $0$ for long request sequences, approaching that correctly predicted by zero-in paging for the offline optimal paging algorithm, by avoiding ``bad'' request sequences (for which zero-in paging correctly predicts poor performance with online memory management). ST paging instead incorrectly predicts that \emph{every} paging algorithm, including the offline optimal, incurs at least one fault per node/solution generated (the cold miss when first accessing $s_i^j$): an abysmal fault rate of over $33\%$, irrespective of the length of the request sequence and of how much memory is available. Thus, the (correct) prediction made by ST paging that in practice one can extract from online paging algorithms such as LRU almost the same performance as from optimal offline paging appears just a fortuitous artifact of the equally abysmal performance ST paging (incorrectly) attributes to all paging algorithms, under all conditions, under every parameterization of the beam search/genetic algorithm.

\section{A robust negative result}
\label{sec:robust}
A common criticism of competitive analysis is that, by focusing on the worst case, it tends to make excessively pessimistic predictions for the performance of online algorithms relative to the offline optimal. For ST paging this is not so true: many simple deterministic online algorithms are guaranteed performance optimal within a factor $2$ in terms of space and time (i.e.\ these algorithms outperform an optimal offline algorithm that runs on a memory system with twice the capacity and half the access cost). Still, several techniques have been developed to refine the blunt approach of competitive analysis, and it is natural to ask whether they can lessen the very negative result of Theorem~\ref{thm:online} (the suboptimality of online algorithms increases from a factor $2$ to an arbitrarily large factor). The answer seems, to a large extent, negative.

Roughly speaking, there are four broad classes of techniques to improve the pessimistic predictions of competitive analysis. The first involves randomization -- but Theorem~\ref{thm:online} already takes it into account. The second involves amortizing the performance evaluation over a spectrum of memory sizes \cite{youngcachechanges}; again, note that in Theorem~\ref{thm:online} the offline algorithm running with just $4$ pages of memory capacity can outperform on $\sigma_m$ every online algorithm running with \emph{any} memory capacity no larger than $m$ pages. The third class of techniques involves excluding ``pathological'' request sequences and/or amortizing the performance evaluation over a spectrum of request sequences~\cite{accessgraph,multifinger,averageorder,bijective}, on two grounds: realistic computations are not tailored against the paging algorithm, and they exhibit some natural properties such as locality. However, the request sequences of Theorem~\ref{thm:online} engender universally poor performance in all online paging algorithms, and no realistic model should exclude them as ``pathological'' since they are representative of algorithmic techniques widely used in the solution of optimization problems. 

A fourth approach is \emph{Full-Cost analysis}~\cite{fullcost} -- correctly accounting for non-fault accesses by assigning them a small but strictly positive cost $\epsilon$. Obviously, in this model every paging algorithm incurs a cost within a factor at most $\frac{1}{\epsilon}$ of any other, regardless of the respective memory capacities, and thus every online algorithm has a competitive ratio of at most $\frac{1}{\epsilon}$. This seems the only approach that can improve the results of Theorem~\ref{thm:online} -- but only very mildly, since one can set the parameters of the beam-search/genetic-algorithm request sequences from Theorem~\ref{thm:online} so that, for every possible level of resource augmentation, the product of competitive ratio and resource augmentation is still $\Omega(\frac{1}{\epsilon})$:

\begin{theorem}
\label{thm:fullcostweak}
Let the cost of a fault be $1$, and that of other requests be $\epsilon>0$. 
Consider an optimal offline paging algorithm OPT with memory $h$ and an online paging algorithm ALG with memory $k\geq h$, and denote by $c_{OPT}^h(\sigma)$ and $c_{ALG}^k(\sigma)$ their respective costs on a request sequence $\sigma$.

Consider the request sequence $\sigma_\ell(m,M)$ (from Section~\ref{sec:online}), making the additional (realistic) assumption that each subsequence $\sigma_i^j$ involves at most $a$ accesses to compute $s_i^j$, for some constant $a$  independent of $\ell$, $m$ and $M$; and in particular choose $m=\lfloor\frac{h-1}{3}\rfloor$ and $M=2k$. Then, for all arbitrarily large $\alpha$, and all sufficiently large $\ell$,
\begin{equation*}
c^k_{ALG}(\sigma_\ell(m,M))>\alpha+
\frac{1-\epsilon}{17a}\cdot \frac{1}{\epsilon}\frac{h}{k}c^h_{OPT}(\sigma_\ell(m,M)).
\end{equation*}

\end{theorem}
\begin{proof}
As in the proof of Lemma~\ref{lem:online}, we can restrict our analysis to deterministic algorithms by invoking Yao's principle. If we can prove that ALG incurs an expected cost at least $\frac{1-\epsilon}{16a}\cdot \frac{1}{\epsilon}\frac{h}{k}$ times that of OPT on every level/generation beyond the first, we immediately obtain the thesis, since both the cost of computing the first generation and the additive cost $\alpha$ can be made an arbitrarily small fraction of the total cost by choosing a sufficiently large $\ell$.

$\sigma_i^j$ involves access to at most $a$ pages. As in Lemma~\ref{lem:offline}, OPT with memory $h\geq 3m+1$ never incurs a fault on any such sequence with $i>1$; the total cost it incurs to compute $\sigma_i^1,\dots,\sigma_i^M$ for $i>1$ is then at most $\epsilon Ma$. ALG pays the same cost, plus an additional cost of $1-\epsilon$ for each fault. From Lemma~\ref{lem:online}, the expected number of such faults with $k=\frac{M}{2}$ is at least $\frac{m}{2}$, leading to a ratio between the costs incurred by ALG and OPT at least equal to
\begin{equation*}
\frac{(1-\epsilon)m/2+\epsilon Ma}{\epsilon Ma}
=1+ \frac{1-\epsilon}{\epsilon}\cdot\frac{m}{2a(M)}
=1+ \frac{1-\epsilon}{\epsilon}\cdot\frac{\lfloor\frac{h-1}{3}\rfloor}{4ak}
\geq \frac{1-\epsilon}{\epsilon}\cdot\frac{h}{16ak},
\end{equation*}
which proves the theorem.
\end{proof}

\section{Tight Full-Cost bounds}
\label{sec:fullcost}
This section obtains a tight characterization of the zero-in $(h,k)-$competitive ratio achievable with a Full-Cost analysis, by first proving a slightly stronger lower bound than that of Theorem~\ref{thm:fullcostweak}, and then a matching upper bound.

The actual sequences we use to prove the lower bound, unlike those in Sections~\ref{sec:online} and~\ref{sec:robust}, are tailored against each individual online algorithm, and in the case of randomized algorithms against the past choices of the algorithm (the adaptive adversary model~\cite{rand}). It is interesting to note that this tailoring -- \emph{unlike what happens in ST paging} -- yields no more than a constant factor worsening of the $(h,k)-$competitive ratio, leaving the lower bound asymptotically unchanged from that of Theorem~\ref{thm:fullcostweak} at $1+\Omega(\frac{1}{\epsilon}\frac{h}{k})$.

It is also interesting to note that the algorithms we use to prove the matching upper bound are none other than the well-known LRU, FIFO, FWF and CLOCK; more precisely all algorithms that are \emph{conservative} according to the  original definition of conservative algorithm in~\cite{youngconservative} (which is slightly different from the more popular one in \cite{competitive} since it considers evictions rather than faults, and thus encompasses FWF and FIFO as well as LRU and CLOCK): 

\begin{definition}
\label{def:conservative}
A paging algorithm is \emph{conservative} if for any request sequence $\sigma$, with memory capacity $k$, it incurs at most $k$ \emph{evictions} on any subsequence of $\sigma$ involving at most $k$ pages.
\end{definition}

Simple proofs that LRU, FIFO, CLOCK and FWF are conservative algorithms according to the definition above can be found in \cite{youngconservative}. We can then prove:

\begin{theorem}
\label{thm:fullcostcompetitive}
Under a Full-Cost analysis, any algorithm that is conservative according to Definition~\ref{def:conservative} (including LRU, FIFO, CLOCK, and FWF) has a zero-in $(h,k)-$competitive ratio $\rho_\emptyset(h,k)=1+\frac{1-\epsilon}{\epsilon}\frac{h-1}{k}$; and this ratio is optimal for online paging algorithms (even if randomized, as long as the request sequence can depend on the past choices of the algorithm).
\end{theorem}

\begin{proof}
We first prove that $\rho_\emptyset(h,k)$ is a lower bound to the $(h,k)-$competitive ratio achievable under zero-in paging (in the Full-Cost model).
Let ALG be a generic online algorithm operating with memory $k$.
Consider a request sequence $\sigma$ formed by subsequences $\sigma_1,\sigma_2,\dots$ each, save at most the last, involving $k$ distinct pages. $\sigma_1$ simply requests, in order, $k$ distinct pages, and these are the only pages that are not zero-in. $\sigma_i$, with $i>1$, is partitioned in up to three phases: the \emph{head}, the \emph{midsection}, and possibly the \emph{tail}.
The head consists of a single request for a ``new'', never requested page. Out of the set formed by this page and by the $k$ distinct pages in $\sigma_{i-1}$, the midsection of $\sigma_i$ keeps requesting a page not in ALG's memory, for as long as it can do so without requesting more than $h-1$ distinct pages of $\sigma_{i-1}$. Finally, the tail of $\sigma_i$ requests, in order, $k-h$ ``new'', never requested pages. 

By construction, ALG experiences a fault at every request in the midsection of $\sigma_i$. On the other hand, the optimal offline algorithm OPT with memory $h$ can easily avoid incurring \emph{any} faults during $\sigma_i$ if, at the end of $\sigma_{i-1}$, it holds in memory the pages in the midsection of $\sigma_i$ (as every other request is the first request for a given zero-in page). OPT can obviously do so for $i=2$; the same is true for $i>2$ since at the end of the midsection of $\sigma_{i-1}$ OPT can hold in memory all the $h$ pages in the head and midsection of $\sigma_{i-1}$, and then replace those not in the midsection of $\sigma_i$ with the corresponding ones in the tail of $\sigma_{i-1}$. 

Note that, if the midsection of $\sigma_i$ can be extended indefinitely, the ratio between the costs incurred by ALG and OPT can be made arbitrarily close to $\frac{1}{\epsilon}$ -- since on every midsection request ALG incurs a fault and OPT does not. Otherwise, denoting by $\ell_i$ the number of (not necessarily distinct) requests in $\sigma_i$, the midsection of $\sigma_i$ contains $(\ell_i-(k-h+1))$ requests and the ratio of the costs incurred by ALG and by OPT while servicing $\sigma_i$ for $i>1$ is
\begin{equation*}
\label{eqn:fullcostlbbasic}
r^{h,k}(\ell_i)=\frac{\epsilon(k-h+1)+(\ell_i-(k-h+1))}{\epsilon\ell_i}.
\end{equation*}

Note that $r^{h,k}(\ell_i)$ for $\epsilon<1$ strictly increases with $\ell_i$. Then, setting $\ell_i$ to its minimum value $k$ ($\sigma_i$ involves at least $k$ requests), we obtain that
\begin{equation*}
\label{eqn:fullcostlbfull}
\frac{c^k_{ALG}(\sigma)}{c^h_{OPT}(\sigma)}
\geq
\frac{(n-1)\epsilon(k-h+1)+(k-(k-h+1))+k}{(n-1)\epsilon k +k},
\end{equation*}
where the last term converges for $n\rightarrow\infty$ to
\begin{equation*}
\label{eqn:fullcostlbfinal}
\frac{\epsilon(k-h+1)+(k-(k-h+1))}{\epsilon k}
=1+\frac{1-\epsilon}{\epsilon}\frac{h-1}{k}=\rho_\emptyset(h,k).
\end{equation*}

We then prove the optimality of conservative algorithms by showing that their
$(h,k)$-competitive ratio is no higher than $1+\frac{1-\epsilon}{\epsilon}\frac{h-1}{k}$.

Although the presence of zero-in pages does considerably complicate the analysis,
we still use the standard technique of partitioning a generic request sequence $\sigma$ into maximal length subsequences $\sigma_1,\dots, \sigma_n$ involving $k$ distinct pages each, save possibly the last . Let $z_i$ be the number of first accesses to zero-in pages during $\sigma_i$. The total number of faults incurred on $\sigma$ by a conservative algorithm ALG with memory $k$ cannot exceed the total number of evictions minus the total number of distinct zero-in pages accessed, and thus is at most
\begin{equation*}
\label{eqn:fullcostevictions}
n\cdot k - \sum_{i=1}^n z_i = \sum_{i=1}^n (k-z_i).
\end{equation*}

The subsequence formed by $\sigma_i$, $i<n$, plus the first request of $\sigma_{i+1}$ involves $k+1$ distinct pages (otherwise $\sigma_i$ would not be maximal). Then, the subsequence extending from immediately after the first request of $\sigma_i$ to immediately after the first request of $\sigma_{i+1}$ involves at least $k$ distinct pages different from the first page of $\sigma_i$; at the beginning of this subsequence an optimal offline algorithm OPT with memory $h$ can have in memory at most $h-1$ of those $k$ pages (since it must also have in memory the first page of $\sigma_i$). Thus, in this interval, OPT incurs a number of faults no less than $\max(0, k-h+1-z_i)$ -- which we shall write as $(k-h+1-z_i)^+$ for brevity. Denoting by $\ell_i$ the number of (not necessarily distinct) requests in $\sigma_i$, the cost  incurred by ALG is then
\begin{equation*}
\label{eqn:fullcostALG}
c^k_{ALG}(\sigma)\leq\sum_{i=1}^n (\epsilon\ell_i + (1-\epsilon)(k-z_i)).
\end{equation*}
The cost incurred by OPT is
\begin{equation*}
\label{eqn:fullcostOPT}
c^h_{OPT}(\sigma)\geq\sum_{i=1}^{n-1} (\epsilon\ell_i + (1-\epsilon)(k-h+1-z_i)^+) + \epsilon \ell_n.
\end{equation*}
Then
\begin{equation*}
\label{eqn:fullcostcompratioprel}
c^k_{ALG}(\sigma)\leq
\frac{\sum_{i=1}^{n-1} (\epsilon\ell_i + (1-\epsilon)(k-z_i)) + \epsilon\ell_n}
{\sum_{i=1}^{n-1} (\epsilon\ell_i + (1-\epsilon)(k-h+1-z_i)^+) + \epsilon\ell_n}
\cdot c^h_{OPT}(\sigma) + (1-\epsilon)(k-z_n).
\end{equation*}
And since the maximum of 
$\frac{\epsilon\ell_i + (1-\epsilon)(k-z_i)}
{\epsilon\ell_i + (1-\epsilon)(k-h+1-z_i)^+}$ is
$1+\frac{1-\epsilon}{\epsilon}\frac{h-1}{k}$ (obtained for $z_i=k-h+1$ and $\ell_i=k$, noting that $\sigma_i$ contains at least $k$ requests) we have that
\begin{equation*}
\label{eqn:eqn:fullcostcompratioub}
c^k_{ALG}(\sigma) < \mleft(1+\frac{1-\epsilon}{\epsilon}\frac{h-1}{k}\mright)\cdot c^h_{OPT}(\sigma) + k.\qedhere
\end{equation*}
\end{proof}

\section{Zero-out paging}
\label{sec:out}
In ST paging, the cost of a fault accounts both for the cost to bring the corresponding page into memory (the actual fault), and for the cost to eventually evict the page. Note that the latter is not significant if the metric under analysis is time, since page eviction can be (and typically is) removed from the critical path through the use of an eviction buffer. However, if the metric is e.g.\ energy, the cost of evicting a page is comparable to that of bringing it into memory. In such a scenario, one might argue that making a non-input page zero-in underestimates the costs involved: although the page starts out in memory ``for free'', the cost of the first eviction should not be ignored.

We can easily extend zero-in paging to deal with this issue accounting separately for the cost of faults and of evictions, and associating to each sequence both a zero-in and a (not necessarily disjoint) \emph{zero-out} page set. Zero-in pages are then, informally, still those starting out directly in memory, and thus do not pay for the first fault. Zero-out pages are instead those pages that need not be written out (i.e.\ evicted from memory) at the end of the computation -- informally, because they are not part of the output or because they are directly pipelined to another process; thus they do not pay for their \emph{last eviction}.

It is crucial to note that this refinement does not compromise our results, and in particular those of Sections~\ref{sec:online} and~\ref{sec:robust}. A beam search with beam width $m$, or a genetic algorithm with $m$ survivors/generation, can be realistically modelled as having at most $m$ nodes/solutions as output -- all other pages are then zero-out. Under this assumption, whatever the relative costs of faults and evictions, Theorems~\ref{thm:online} and~\ref{thm:fullcostweak} continue to hold virtually unchanged save for a minor and straightforward adaptation of the respective proofs.

\section{Conclusions}
\label{sec:end}
Zero-in paging analysis raises doubts on the performance predictions provided by competitive analysis for many popular online paging algorithms. Note that under ST paging competitive analysis \emph{does} predict to a large extent the good performance exhibited in practice by these algorithms
compared to the offline optimal;
but in the light of our results these correct predictions appear just a fortuitous artifact of an incorrect assumption.

We challenge the algorithms community to find a way to mitigate these very negative results -- namely, that online paging algorithms with proper accounting of faults can have arbitrarily worse performance than the offline optimal even if given access to memory of arbitrarily larger capacity -- without resorting to incorrect assumptions. The task appears difficult with existing tools, in particular because our results are not based on pathological sequences tailored against the individual paging algorithm, but instead on request sequences characteristic of widely used computational techniques (beam searches and genetic algorithms) that engender universally poor performance in all online algorithms. Only a Full-Cost analysis, where the cost ratio between a fault and non-fault access is a finite value $\frac{1}{\epsilon}$, can achieve some success in this direction; but it still predicts that the product of the cost and capacity overheads of any online algorithms compared to the offline optimal (i.e. the product of its competitive ratio and resource augmentation) is at least of the same order as  $\frac{1}{\epsilon}$.

Ultimately, good performance of paging algorithms facing blindly an adversarially crafted request sequence may simply be a mirage. After all, programmers and compilers go to great lengths to exploit their knowledge of the task at hand and of a system's architecture so as to massage memory accesses into request sequences that are as ``nice'' as possible, rather than adversarial. It would not then be surprising if such help were in many cases indispensable for decent  memory management performance. Whether a simple model can capture the quintessential elements of this assistance may well be the most important question this work leaves open.

One consequence of our results is in the popular area of Cache-Oblivious~\cite{cacheoblivious} and Cache-Adaptive~\cite{cacheadaptive} algorithms, where the cache model is an \emph{Ideal Cache} employing an optimal offline replacement policy. The justification provided for this generally unrealistic assumption is the factor-2-optimality of online (ST) paging -- so that switching from an Ideal Cache to a cache of twice the capacity with a realistic replacement policy such as LRU involves at most doubling the total access cost. But this is not true with zero-in paging, requiring a careful re-evaluation of optimality claims made on individual Cache-Oblivious and Cache-Adaptive algorithms: in principle one such algorithm could vastly outperform another on an Ideal Cache and still be vastly outperformed by the latter on a realistic cache employing e.g. LRU replacement.

Finally, it is worth noting that our tight characterization of the zero-in $(h,k)-$competitive ratio under a Full-Cost analysis shows that classic paging algorithms such as LRU and FIFO are still optimal among online paging algorithms. This result does increase our confidence that their comparative effectiveness is not an artifact of the framework, either theoretical or experimental, used to analyse them.

\paragraph*{Acknowledgments.} The authors thank Gianfranco Bilardi for his encouragement and for a review of an early version of this work.

\bibliographystyle{abbrv}
\bibliography{zip}

\end{document}